\documentclass[reqno, 11pt]{amsart}
\usepackage{amscd}

\usepackage[top=4cm, bottom=3cm, left=3.2cm, right=3.2cm]{geometry}

\newtheorem{theorem}{Theorem}[section]
\newtheorem{corollary}{Corollary}

\newtheorem{lemma}[theorem]{Lemma}

\theoremstyle{definition}
\newtheorem{definition}[theorem]{Definition}
\newtheorem{remark}{Remark}

\title[A Collision-Avoiding flocking system]{Mean-Field Limit for a Collision-Avoiding Flocking System and the Time-Asymptotic Flocking Dynamics for the Kinetic Equation}

\author{Rong Yang}
\address{Department of Mathematical Sciences, Tsinghua University,
Beijing, 100084, People's Republic of China}
\email{ysihan2010@163.com}

\author{Li Chen}
\address{Department of Mathematical Sciences, Tsinghua University,
Beijing, 100084, People's Republic of China}
\email{lchen@math.tsinghua.edu.cn}

\thanks{This work is supported by the National
Natural Science Foundation of China (NSFC), No. 11271218. }

\date{}
\begin{document}
\maketitle

\begin{abstract}
A Collision-Avoiding flocking particle system proposed in \cite{FJ} is studied in this paper. The global wellposedness of its corresponding Vlasov-type kinetic equation is proved. As a corollary of the global stability result, the mean field limit of the particle system is obtained. Furthermore, the time-asymptotic flocking behavior of the solution to the kinetic equation is also derived. The technics used for local wellposedness and stability follow from similar ideas to those have been used in \cite{WELL,N77,D79}. While in order to extend the local result globally, the main contribution here is to generate a series of new estimates for this Vlasov type equation, which imply that the growing of the characteristics can be controlled globally. Further estimates also show the long time flocking phenomena.
\end{abstract}


\section{Introduction}

\subsection{Background}
Collective self-driven synchronized motion in multi-agent interactions appears in many applications in biology \cite{BB,D1,D2,D3,JK,JS,CM1,CM2, V1}, ecology, mobile network \cite{MA, DH, HL}, control theory \cite{ID,LP}, sociology and economics \cite{AC,AD}. In the last few years,  two strategies have been used to describe these phenomena in applied mathematics literature: particle dynamics \cite{FS1,FS2,MR, HL} and continuum models for mesoscopic or macroscopic quantities \cite{YL,SD,SE,CM1,CM2}. For instance, in the works of Cucker and Smale \cite{FS1,FS2}, a particle model for the flocking of birds and the flocking behavior with spatial communication rate between autonomous agents were postulated. In the works by Ha and Tadmor \cite{SE}, Carrillo {\it et al.} \cite{JA}, a statistical description of the interacting agent system based on mesoscopic models by means of kinetic equations were studied. An in depth study of kinetic models is very important in understanding the phenomena since this class of models play a role of bridge between the particle models and the macroscopic models.

In this paper, we will focus on a general collision-avoiding flocking system proposed by Felipe Cucker and Jiu-Gang Dong \cite{FJ}. It is a new simple dynamical system which describes the emergency of flocking, especially the behavior of collision-avoidance. It depends on three given functions, the repelling forces $F$ (responsible of avoiding collisions), the coupling forces $G$ (steering towards alignment) and the weighting for the strength of these couplings $\Phi$ (interaction rate). More precisely, the collision-avoiding particle system is given by
\begin{equation}\label{ode1equs}
  \left\{\begin{array}{l}
  x_i'(t)=v_i(t),\\
  v_i'(t)=\dfrac{1}{N}\sum\limits_{j=1}^N \Phi(|x_i-x_j|)G(v_j-v_i)+\dfrac{\Lambda(v)^{2\alpha-1}}{N}\sum\limits_{j=1}^NF(|x_i-x_j|^2)(x_i-x_j),
  \end{array}\right.
\end{equation}
where
$\Lambda(v)=\frac{1}{N}(\sum\limits_{i>j}(|v_i-v_j|^2))^{\frac{1}{2}}$ is the alignment measure and $1\leq\alpha<3/2$. In the case of $G(v)=v$ and $F=0$, the model goes back to Cucker-Smale model in \cite{FS1,FS2}.

The flocking and collision-avoiding results to the complex system \eqref{ode1equs} for fixed number of particles were studied in \cite{FJ}. Especially the collision-avoiding result was obtained under a very restrictive assumption on $F$. However, the number of particles is very large in real application, which implies that the numerical simulation for the corresponding ODE system is difficult, or sometimes impossible. A way of understanding the large particle system is to find the mean-field limit of it. We will rigorously derive a kinetic description for the particle model by using particle method which has been widely used in the literature (for example, \cite{D79,N77,SJ,WELL}), instead of BBGKY hierarchy which has been used in \cite{SE} (for Cucker-Smale model) and study the long time flocking phenomenon of its solution. Moreover, we give an \emph{a priori} $L^\infty$ estimate on the solution of \eqref{kineticequs}, which helps the understanding of collision-avoiding. But the general collision-avoiding result for arbitrary time for the kinetic equations is still open. We expect to have an $L^\infty$ estimate for the particle density $\rho(t,x)=\int_{\mathbb{R}^d}f(t,x,dv)$ to describe no aggregation of particles in time. This will be one of the future topics to study.

Now, we start with presenting a well-known approach to the wellposedness (existence, uniqueness and stability) of the kinetic model, where some basic knowledge of optimal transport theory \cite{TR} will be used. This method was used in \cite{WELL} aiming to give a global wellposedness result for a general system. However, there were no detailed estimates on growing of the characteristics so that one is not so sure whether the local existence could be extended globally in time, see remark \ref{remonref3}. In this article, we give a clear estimate on the compact support of the solution for any given time $t$ for this specific model, which is the main contribution here.

Let $f=f(t,x,v)$ be the one-particle distribution of such particles positioned at $(t,x)\in\mathbb{R}_+\times\mathbb{R}^d$ with a velocity $ v\in\mathbb{R}^d$. The corresponding kinetic equation of \eqref{ode1equs} in the sense of mean field limit is
\begin{equation}\label{kineticequs}
\partial_tf+v\cdot\nabla_x f +\nabla_v\cdot\Big(H_{[f]}(t,x,v)f\Big)=0,
\end{equation}
where
$\rho(t,x)=\displaystyle\int_{\mathbb{R}^{d}}f(t,x,dv)$ and
\begin{align*}
H_{[f]}(t,x,v)=&-\big[\Phi(|x|)G(v)\big]*f\\
&\quad +\Big(\int_{\mathbb{R}^{2d}}\!\!|v|^2f(dx,dv)-\big(\int_{\mathbb{R}^{2d}}\!\!vf(dx,dv)\big)^2\Big)^\frac{2\alpha-1}{2}[F(|x|^2)x]*\rho(t,x),
\end{align*}
where $\big[\Phi(|x|)G(v)\big]*f=\int_{\mathbb{R}^{2d}} \Phi(|x-y|)G(v-w)f(dy,dw)$, and $[F(|x|^2)x]*\rho(t,x)=\int_{\mathbb{R}^{d}}F(|x-y|^2)(x-y)\rho(t,dy)$.\\
The second goal of this paper is to prove that the kinetic model \eqref{kineticequs} exhibits time-asymptotic flocking behavior when the interaction rate has a uniformly non-negative lower bound.

\subsection{Structure of the paper}
The paper is organized as follows. In section 2 we give the motivation for the definition of measure valued solution of \eqref{kineticequs} (using push forward of measure)  and the main results of this paper. In section 3, the usual argument by using Banach fixed point theorem with Monge-Kantorovich-Rubinstein distance is applied in getting the local wellposedness result. Then after that, in section 4, we give a series of uniform estimates, including the boundedness of second moments in both position and velocity, in getting the control of the increasing of the compact support of the solution. In section 5, as a byproduct of stability, we show that how to get mean field limit of \eqref{ode1equs} by using particle method. The last section will be focused on the long time behavior of the solution, mainly on the flocking behavior. Namely, the kinetic velocity fluctuation decay in time and the kinetic velocity fluctuation is bounded in time.

\section{Definition of the solution and the main results}

\subsection{Definition of the measure valued solution}

In \eqref{kineticequs}, the total mass is preserved, which means that we can normalize the equation and consider only the solutions with total mass one. Therefore we can reduce ourselves to work with probability measures.

The idea of the definition of measure valued solution is the following. Notice that the associated characteristic system of \eqref{kineticequs} is
\begin{equation}\label{ode2equs}
  \left\{\begin{array}{l}
  \frac{d}{dt}X=V,\\
  \frac{d}{dt}V=H_{[f]}(t,X,V).
  \end{array}\right.
\end{equation}
If for any given initial data $(x,v)\in\mathbb{R}^{d}\times\mathbb{R}^{d}$, \eqref{ode2equs} has a solution $(X,V)$, we denote the flow of the equation \eqref{ode2equs} at time $t\in[0,T]$ (for some $0<T\leq +\infty$) by
\begin{eqnarray}\label{defmathcalT}
&{\mathcal{T}}^t_{H_{[f]}}: &\mathbb{R}^{d}\times\mathbb{R}^{d}\rightarrow\mathbb{R}^{d}\times\mathbb{R}^{d},\\
&&(t,x,v)\mapsto{\mathcal{T}}^t_{H_{[f]}}(x,v)=(X,V),
\end{eqnarray}
or equivalently the map $(t,x,v)\mapsto{\mathcal{T}}^t_{H_{[f]}}(x,v)$
is continuous, then we can define the push-forward of any given probability measure $f_0(x,v)$ by $f(t,x,v)={\mathcal{T}}^t_{H_{[f]}}\#f_0$. Here, push-forward means that for any Borel measurable set $\Omega \subset\mathbb{R}^{d}\times\mathbb{R}^{d}$, and $t\in[0,T]$
$$
\int_{\mathcal{T}^t_{H_{[f]}}(\Omega)}f_0(dx,dv)=\int_{\Omega}f(t,dx,dv),
$$
then equation
$$
f(t,x,v)={\mathcal{T}}^t_H{_{[f]}}\#f_0
$$
can be viewed as a weak formulation of kinetic equation \eqref{kineticequs}.
Actually, $f(t,x,v)={\mathcal{T}}^t_H{_{[f]}}\#f_0$ has another form in the sense of distribution, i.e. $\forall\varphi\in C_0^1([0, T]\times\mathbb{R}^{d}\times\mathbb{R}^{d})$,
\begin{equation}\label{weaksystem}
\frac{d}{dt}\langle f(t,x,v),\varphi\rangle=\langle f(t,x,v),\partial_t\varphi+v\cdot\nabla_x\varphi+H_{[f]}\cdot\nabla_v\varphi\rangle,
\end{equation}
where $\langle f(t,x,v),\varphi \rangle=\int_{\mathbb{R}^{2d}}\varphi(t,x,v)f(t,dx,dv)$.

Due to the typical structure of kinetic equation \eqref{kineticequs}, the following space is needed
\begin{align*}
\mathcal{P}^*_1(\mathbb{R}^{d}\times\mathbb{R}^{d})=&\Big\{f|~f \mbox{ is a probability measure in } \mathbb{R}^{d}\times\mathbb{R}^{d}, \mbox{ and }\\
&\int_{\mathbb{R}^{2d}}\big(|x|+|v|+|v|^2\big)f(dx,dv)
<+\infty\Big\}.
\end{align*}

\begin{definition}\label{defsolution}
Given an initial measure $f_0\in\mathcal{P}^*_1\big(\mathbb{R}^{d}\times\mathbb{R}^{d}\big)$,
for any $T\in(0,\infty]$, $f\in C\big([0, T];\mathcal{P}^*_1(\mathbb{R}^{d}\times\mathbb{R}^{d})\big)$ is called a measure valued solution of \eqref{kineticequs} with initial data $f_0$, if
$$
f(t,x,v)={\mathcal{T}}^t_{H_{[f]}}\#f_0.
$$
\end{definition}

\subsection{Main results}
We will give the following results on the existence, uniqueness, stability and the long time behavior of the measure valued solution.

For convenience, we list the assumptions that we are going to use in this paper.
{\bf Assumptions} Assume that
\begin{enumerate}
\item \label{assum1} $\Phi(|x|)$, $G(v)$, $F(|x|^2)$ are locally Lipschitz.
\item \label{assum2}The coupling function $G: \mathbb{R}^d\rightarrow\mathbb{R}^d$ satisfies:
\begin{enumerate}
  \item  $G(v)=-G(-v)$, $\forall v\in\mathbb{R}^d$;
  \item  $G(v)\cdot v \geq G^* |v|^{2\alpha}$, $\forall v\in\mathbb{R}^{d}$ with $\alpha\in [1,\frac{3}{2})$ and $G^*>0$;
\end{enumerate}
\item \label{assum3}$\Phi(|x|)\geq \Phi^*>0$, $|F(|x|^2)x|\leq F^*$ and $F^*<2^{\alpha-\frac{1}{2}} \Phi^* G^*$;
\item \label{assum4} $\forall x\in\mathbb{R}^d$, $|\Phi(|x|)G(v)|\leq C(1+|x|+|v|) $.
\end{enumerate}

Here we mention that assumption \eqref{assum1} is necessary to study the complex system \eqref{ode1equs}, or equivalently, the characteristic system for \eqref{kineticequs}. Assumption \eqref{assum2} is a natural symmetric and coercive condition for the coupling force, one typical example is $G(v)=v$. Assumption \eqref{assum3} implies that the force generated by collision-avoiding is weaker than the flocking driven force. Assumption \eqref{assum4} is needed in the control of the growth of characteristics.

\begin{theorem}\label{T1}[Global wellposedness and long time behavior]
Under the assumptions \eqref{assum1}-\eqref{assum4}, $\alpha\in[1,5/4)$, if the initial probability measure $f_0$ has a compact support in $B_{R_0}$, then there exists a unique global in time measure valued solution $f$ to equation \eqref{kineticequs}.
Moreover, the following long time behaviors of the solution hold.
\begin{enumerate}
\item There is an increasing in time radius $R(R_0,t)>0$ such that,
\begin{align*}
{\rm supp}f(t,x,v)\subset B_{R(R_0,t)}\subset\mathbb{R}^{d}\times\mathbb{R}^{d}, \mbox{ for all } t\in[0,+\infty);
\end{align*}
\item ``Long time flocking phenomena''
\begin{align*}
\displaystyle\lim_{t\rightarrow +\infty}\int_{\mathbb{R}^{2d}}|v-\mathcal{V}_1(0)|^2f(t,dx,dv)=0,
\end{align*}
where $\mathcal{V}_1(0)$ is the initial group velocity $\displaystyle\int_{\mathbb{R}^{2d}}vf_0(dx,dv)$, and
\begin{align*}
\sup\limits_{0\leq t<\infty}\Big(\int_{\mathbb{R}^{2d}}|x|^2f(dx,dv)-\big(\int_{\mathbb{R}^{2d}}xf(dx,dv)\big)^2\Big)<C_{R_0}.
\end{align*}
\item ``Long time asympototics''
\begin{align*}
\displaystyle\lim_{t\rightarrow +\infty} \mathcal{W}_1(f(t,x,v),\rho(t,x)\delta (v-\mathcal{V}_1(0))) =0.
\end{align*}
where $\rho(t,x)=\int_{\mathbb{R}^{d}}f(t,x,dv)$.
\end{enumerate}
\end{theorem}

This global result will be proved by a combination of local wellposedness and global uniform a priori estimates on the growth of compact support.
Due to the typical structure of the equation \eqref{kineticequs}, there are two important quantities, namely
\begin{align*}
\mathcal {G}[f(t)]&=\int_{\mathbb{R}^{2d}}|v|^2f(dx,dv)-\big(\int_{\mathbb{R}^{2d}}vf(dx,dv)\big)^2,\\
\Gamma[f(t)]&=\int_{\mathbb{R}^{2d}}|x|^2f(dx,dv)-\big(\int_{\mathbb{R}^{2d}}xf(dx,dv)\big)^2.
\end{align*}
It is the decay in time property of $\mathcal{G}[f(t)]$ (for $\alpha\in [1,3/2)$) and the uniform in time boundedness of $\Gamma[f(t)]$ (for $\alpha\in [1,5/4)$) that makes us possible to get the time dependent control of the growth of characteristics. Due to technical reasons, we can not handle the case $\alpha\in [5/4,3/2)$.

\begin{theorem}\label{T2}[Stability and mean field limit]
Under the same assumptions as in theorem \ref{T1}, let $f$, $g$ be two solutions of equation \eqref{kineticequs} with initial data $f_0$, $g_0$ with compact supports respectively, then there exists an increasing smooth function $\lambda(t)$ depending only on the size of the supports of $f_0$ and $g_0$, such that
$$
\mathcal {W}_1(f,g)\leq\lambda(t)\mathcal {W}_1(f_0,g_0).
$$
Consequently, \eqref{kineticequs} is a mean field limit of \eqref{ode1equs}.
\end{theorem}

\begin{remark}
The stability result is obtained by using $\mathcal {W}_1$, the Monge-Kantorovich-Rubinstein distance. For the reader's convenience, we will list the definition of it in the appendix.
\end{remark}

\section{Local existence and uniqueness}

\begin{theorem}
\label{thmlocal}
Under the assumption \eqref{assum1}, if the initial probability measure $f_0$ has compact support in $B_{R_0}$ for $R_0>0$, then there exists $T>0$ such that equation \eqref{kineticequs} has a unique measure valued solution $f$ and ${\rm supp}f(t,\cdot,\cdot)\subset B_{2R_0}$, $\forall t\leq T$. 
\end{theorem}

In order to prove theorem \ref{thmlocal}, we need the following estimates for the characteristic equations \eqref{kineticequs}.
Let $\mathcal{P}(\mathbb{R}^{d}\times\mathbb{R}^{d})$  be the space of probability measures.
\begin{lemma} \label{Hlp}
If Assumption \eqref{assum1} holds, then for any $f\in C\big([0, T]; \mathcal{P}(\mathbb{R}^{d}\times\mathbb{R}^{d})\big)$ with compact support in $B_R$ for some positive $R$ depending on $T$, the following hold
\begin{enumerate}
  \item $|H_{[f]}(t,x,v)|\leq C_R$, for all $(t,x,v)\in [0,T]\times B_R$;
   \item $H_{[f]}(t,x,v)$ is continuous in $[0,T]\times\mathbb{R}^{d}\times\mathbb{R}^{d}$;
  \item $H_{[f]}(t,x,v)$ is locally Lipschitz with respect to x, v.
i.e. for all bounded $\Omega\subset \mathbb{R}^{d}\times\mathbb{R}^{d}$ there exists $L_{R,\Omega}>0$, such that for $t\in[0,T],(x_1,v_1),(x_2,v_2)\in \Omega$,
\begin{align*}
\big|H_{[f]}(t,x_1,v_1)-H_{[f]}(t,x_2,v_2)\big|\leq{L_{R,\Omega}\big|(x_1,v_1)-(x_2,v_2)\big|}.
\end{align*}
\end{enumerate}
Especially, if $\Omega=B_R$, we use constant $L_R$ instead of $L_{R,\Omega}$ for convenience.
\end{lemma}

\begin{proof}
We only need to prove that $H_{[f]}(t,x,v)$ is locally Lipschitz, the other statements are trivial. By using the local Lipschitz continuity of $\Phi(|x|)$, $G(v)$, $F(|x|^2)$, i.e. Assumption \eqref{assum1}, we have
\begin{align*}
&\big|H_{[f]}(t,x_1,v_1)-H_{[f]}(t,x_2,v_2)\big|\\
=&\Big|\int_{\mathbb{R}^{2d}}\hspace{-3mm}\Phi(|x_1-y|)G(v_1-w)f(t,dy,dw)\\
&-\int_{\mathbb{R}^{2d}}\hspace{-3mm}\Phi(|x_2-y|)G(v_2-w)f(t,dy,dw)\\
&+\Big(\int_{\mathbb{R}^{d}} \hspace{-3mm} F(|x_1-y|^2)(x_1-y)\rho(t,dy)\Big)\Big(\int_{\mathbb{R}^{2d}}\hspace{-3mm}|v|^2f(dx,dv) -\big(\int_{\mathbb{R}^{2d}}\hspace{-3mm}vf(dx,dv)\big)^2\Big)^\frac{2\alpha-1}{2}\\
&-\Big(\int_{\mathbb{R}^{d}} \hspace{-3mm} F(|x_2-y|^2)(x_2-y)\rho(t,dy)\Big)\Big(\int_{\mathbb{R}^{2d}}\hspace{-3mm}|v|^2f(dx,dv) -\big(\int_{\mathbb{R}^{2d}}\hspace{-3mm}vf(dx,dv)\big)^2\Big)^\frac{2\alpha-1}{2}
\Big|\\
\leq&\Big|\int_{\mathbb{R}^{2d}}\hspace{-3mm}\big(\Phi(|x_1-y|)G(v_1-w)-\Phi(|x_2-y|)G(v_2-w)\big)f(t,dy,dw)\Big|\\
&+\Big|\int_{\mathbb{R}^{d}} \hspace{-3mm}\Big(F(|x_1-y|^2)(x_1-y)-F(|x_2-y|^2)(x_2-y)\Big)\rho(t,dy)\\
& \quad \times\Big(\int_{\mathbb{R}^{2d}}\hspace{-3mm}|v|^2f(dx,dv)-\big(\int_{\mathbb{R}^{2d}}\hspace{-3mm}vf(dx,dv)\big)^2\Big)^\frac{2\alpha-1}{2}\Big|\\
\leq& L_{R,\Omega}\big(|x_1-x_2\big|+\big|v_1-v_2|\big).
\end{align*}
\end{proof}

 We can directly obtain the following Corollary from Lemma \ref{Hlp}.
\begin{corollary} \label{Fl}
If Assumption \eqref{assum1} holds, then for any $f, g\in C\big([0, T]; \mathcal{P}(\mathbb{R}^{d}\times\mathbb{R}^{d})\big)$ with compact support in $B_R$ for some positive $R$ depending on $T$, the following statements for $\Psi_{H_{[f]}}=\big(v,H_{[f]}(t,x,v)\big)$ hold
\begin{itemize}
  \item[(i )]There exists $C_R>0$ such that
      $$\big|\Psi_{H_{[f]}}\big|=\big|\big(v,H_{[f]}(t,x,v)\big)\big|\leq{C_R} \mbox{ for all }(t,x,v)\in[0,T]\times{B_R};$$
  \item[(ii)]$\Psi_{H_{[f]}}$ is locally Lipschitz with respect to $x,v$. Especially, $\forall (x_1,v_1),(x_2,v_2)\in B_R$,
  $$\big|\Psi_{H_{[f]}}(t,x_1,v_1)-\Psi_{H_{[f]}}(t,x_2,v_2)\big|\leq\big(1+L_R\big)\big|(x_1,v_1)-(x_2,v_2)\big|;$$
  \item[(iii)]For any compact set $\Omega$, we have
$$
\big\|\Psi_{H_{[f]}}-\Psi_{H_{[g]}}\big\|_{L^{\infty}(\Omega)}=\big\|H_{[f]}-H_{[g]}\big\|_{L^{\infty}(\Omega)} .
$$
\end{itemize}
\end{corollary}


\begin{lemma} \label{p1}
If Assumption \eqref{assum1} holds, then for any $f, g\in C\big([0, T]; \mathcal{P}(\mathbb{R}^{d}\times\mathbb{R}^{d})\big)$ with compact support in $B_R$ for some positive $R$ , and for any $(x_0,v_0)\in B_{R_0}$ such that $\forall t\in [0,T]$, ${\mathcal{T}}^t_{H_{[f]}}(x_0,v_0)$,
${\mathcal{T}}^t_{H_{[g]}}(x_0,v_0)$ $\in{B_R}$ for some $R>0$, there is a constant $C_R$ such that
\begin{align*}
\big\|{\mathcal{T}}^t_{H_{[f]}}-{\mathcal{T}}^t_{H_{[g]}}\big\|_{L^{\infty}(B_{R_0})}\leq\frac{e^{C_Rt}-1}{C_R}
\sup\limits_{t\in{[0,T]}}\|H_{[f]}-H_{[g]}\|_{L^{\infty}(B_R)}\quad \qquad t\in[0,T].
\end{align*}
\end{lemma}

\begin{proof} By definition, we know that $\mathcal{T}^t_{H_{[f]}}(x_0,v_0)$ and $\mathcal{T}^t_{H_{[g]}}(x_0,v_0)$are the solutions of the ODE systems with initial data $(x_0,v_0)$, i.e.
\begin{align*}
\frac{d}{dt}\mathcal{T}^t_{H_{[f]}}(x_0,v_0)=\Psi_{H_{[f]}}\big(t,x_1(t),v_1(t)\big), \quad
\frac{d}{dt}\mathcal{T}^t_{H_{[g]}}(x_0,v_0)=\Psi_{H_{[g]}}\big(t,x_2(t),v_2(t)\big).
\end{align*}
Therefore, by rewriting the above ODEs into integral form and using Corollary \ref{Fl}, we have
\begin{align*}
&\big|\mathcal{T}^t_{H_{[f]}}(x_0,v_0)-\mathcal{T}^t_{H_{[g]}}(x_0,v_0)\big|\\
=&\Big|\int_{0}^t(\Psi_{H_{[f]}}(s,x_1(s),v_1(s))-\Psi_{H_{[g]}}(s,x_2(s),v_2(s)))ds\Big| \\
\leq&\int_{0}^t\big|\Psi_{H_{[f]}}(s,x_1(s),v_1(s))-\Psi_{H_{[g]}}(s,x_2(s),v_2(s))\big|ds \\
\leq&\int_{0}^t\big|\Psi_{H_{[f]}}(s,x_1(s),v_1(s))-\Psi_{H_{[f]}}(s,x_2(s),v_2(s))\big|ds \\
&+\int_{0}^t\big|\Psi_{H_{[f]}}(s,x_2(s),v_2(s))-\Psi_{H_{[g]}}(s,x_2(s),v_2(s))\big|ds\\
\leq&\big(1+L_R(H_{[f]})\big)\int_{0}^t\big|(x_1(s),v_1(s))-(x_2(s),v_2(s))\big|ds\\
&+\int_{0}^t\|H_{[f]}-H_{[g]}\|_{L^{\infty}(B_R)}ds\\
=&\big(1+L_R(H_{[f]})\big)\int_{0}^t\big|\mathcal{T}^t_{H_{[f]}}(x_0,v_0)-\mathcal{T}^t_{H_{[g]}}(x_0,v_0)\big|ds\\
&+\int_{0}^t\|H_{[f]}-H_{[g]}\|_{L^{\infty}(B_R)}ds,
\end{align*}
where $L_R(H_{[f]})$ is Lipschitz constant of $H_{[f]}$ in the ball $B_R$.
Thus Gronwall's Lemma implies
\begin{align*}
\big|\mathcal{T}^t_{H_{[f]}}(x_0,v_0)-\mathcal{T}^t_{H_{[g]}}(x_0,v_0)\big|
\leq&\int_{0}^te^{\big(1+L_R(H_{[f]})\big)(t-s)}\|H_{[f]}-H_{[g]}\|_{L^{\infty}(B_R)}ds\\
\leq&\frac{e^{C_Rt}-1}{C_R}\sup\limits_{t\in{[0,T]}}\|H_{[f]}-H_{[g]}\|_{L^{\infty}(B_R)},
\end{align*}
where $C_R=1+L_R(H_{[f]})$.
\end{proof}

\begin{lemma} \label{l5}
If Assumption \eqref{assum1} holds, then for any $f$, $g\in C\big([0£¬+\infty);~\mathcal{P}(\mathbb{R}^{d}\times\mathbb{R}^{d})\big)$ with compact supports in $B_R$ for some $R>0$, there exists a constant $C_R$ such that
\begin{align*}
\big\|H_{[f]}-H_{[g]}\big\|_{L^{\infty}({B_R)}}\leq C_R\mathcal{W}_1(f,g).
\end{align*}
\end{lemma}

\begin{proof}
For any $(x,v)\in{B_R}$, we have
\begin{align*}
\begin{split}
~&\big|H_{[f]}(x,v)-H_{[g]}(x,v)\big|\\
\leq&\Big|\int_{\mathbb{R}^{2d}}\hspace{-3mm}\Phi(|x-y|)G(v-w)f(t,dy,dw)-\int_{\mathbb{R}^{2d}}\hspace{-3mm}\Phi(|x-y|)G(v-w)g(t,dy,dw)\Big|\\
&+\Big|\Big(\int_{\mathbb{R}^{d}}\hspace{-3mm}F(|x-y|^2)(x-y)\rho[f](t,dy)\Big)\Big(\int_{\mathbb{R}^{2d}}\hspace{-3mm}|v|^2f(dx,dv)- \big(\int_{\mathbb{R}^{2d}}\hspace{-3mm}vf(dx,dv)\big)^2\Big)^\frac{2\alpha-1}{2}\\
&-\Big(\int_{\mathbb{R}^{d}}\hspace{-3mm}F(|x-y|^2)(x-y)\rho[g](t,dy)\Big)\Big(\int_{\mathbb{R}^{2d}}\hspace{-3mm}|v|^2f(dx,dv)- \big(\int_{\mathbb{R}^{2d}}\hspace{-3mm}vf(dx,dv)\big)^2\Big)^\frac{2\alpha-1}{2}\Big|\\
\end{split}
\end{align*}
\begin{align*}
\begin{split}
&+\Big|\Big(\int_{\mathbb{R}^{d}}\hspace{-3mm}F(|x-y|^2)(x-y)\rho[g](t,dy)\Big)\Big(\int_{\mathbb{R}^{2d}}\hspace{-3mm}|v|^2f(dx,dv)- \big(\int_{\mathbb{R}^{2d}}\hspace{-3mm}vf(dx,dv)\big)^2\Big)^\frac{2\alpha-1}{2}\\
&-\Big(\int_{\mathbb{R}^{d}}\hspace{-3mm}F(|x-y|^2)(x-y)\rho[g](t,dy)\Big)
\Big(\int_{\mathbb{R}^{2d}}\hspace{-3mm}|v|^2g(dx,dv)-\big(\int_{\mathbb{R}^{2d}}\hspace{-3mm}vg(dx,dv)\big)^2\Big)^\frac{2\alpha-1}{2}\Big|\\
=& I_1+I_2+I_3.
\end{split}
\end{align*}
Let $\pi$ be an optimal transportation plan between the measures $f$ and $g$. Therefore it has compact support in $B_R\times B_R$. Then we can estimate the above quantities term by term in the following,
\begin{align*}
\begin{split}
I_1
:=&\Big|\int_{\mathbb{R}^{2d}}\hspace{-3mm}\Phi(|x-y|)G(v-w)f(t,dy,dw)
-\int_{\mathbb{R}^{2d}}\hspace{-3mm}\Phi(|x-z|)G(v-u)g(t,dz,du)\Big|\\
=&\Big|\int_{\mathbb{R}^{4d}}\hspace{-3mm}\Phi(|x-y|)G(v-w)-\Phi(|x-z|)G(v-u)d\pi(y,w,z,u)\Big|\\
\leq&C_R \int_{\mathbb{R}^{4d}}\hspace{-3mm}\big(|y-z|+|w-u|\big)d\pi(y,w,z,u)\\
\leq&C_R \mathcal{W}_1(f,g),
\end{split}
\end{align*}
\begin{align*}
\begin{split}
I_2
:=&\Big|\int_{\mathbb{R}^{d}}\hspace{-3mm}F(|x-y|^2)(x-y)\rho[f](t,dy)
-\int_{\mathbb{R}^{d}}\hspace{-3mm} F(|x-y|^2)(x-y)\rho[g](t,dy)\Big|\\
&\times \Big|\int_{\mathbb{R}^{2d}}\hspace{-3mm}|v|^2f(dx,dv)-\big(\int_{\mathbb{R}^{2d}}\hspace{-3mm}vf(dx,dv)\big)^2\Big|^\frac{2\alpha-1}{2}\\
\leq&C_R \mathcal{W}_1(f,g),
\end{split}
\end{align*}
and
\begin{align*}
\begin{split}
I_3
:=&\Big|\int_{\mathbb{R}^{d}} \hspace{-3mm} F(|x-y|^2)(x-y)\rho[g](t,dy)\Big|\Big|\Big(\int_{\mathbb{R}^{2d}}\hspace{-3mm}|v|^2f(dx,dv) -\big(\int_{\mathbb{R}^{2d}}\hspace{-3mm}vf(dx,dv)\big)^2\Big)^\frac{2\alpha-1}{2}\\
&-\Big(\int_{\mathbb{R}^{2d}}\hspace{-3mm}|v|^2g(dx,dv)-\big(\int_{\mathbb{R}^{2d}}\hspace{-3mm}vg(dx,dv)\big)^2\Big)^\frac{2\alpha-1}{2}\Big|\\
\leq&C_R \Big|\int_{\mathbb{R}^{2d}}\hspace{-3mm}|v|^2f(dx,dv)-\big(\int_{\mathbb{R}^{2d}}\hspace{-3mm}vf(dx,dv)\big)^2-\int_{\mathbb{R}^{2d}}\hspace{-3mm}|v|^2g(dx,dv) +\big(\int_{\mathbb{R}^{2d}}\hspace{-3mm}vg(dx,dv)\big)^2\Big|,
\end{split}
\end{align*}
thus
\begin{align*}
\begin{split}
I_3
\leq&C_R \Big[\Big|\int_{\mathbb{R}^{2d}}\hspace{-3mm}|v|^2f(t,dx,dv)-\int_{\mathbb{R}^{2d}}\hspace{-3mm}|u|^2g(t,dy,du)\Big| +\Big|\int_{\mathbb{R}^{2d}}\hspace{-3mm}vf(dx,dv) \\
&+\int_{\mathbb{R}^{2d}}\hspace{-3mm}ug(dy,du)\Big|\cdot \Big|\int_{\mathbb{R}^{2d}}\hspace{-3mm}vf(dx,dv)-\int_{\mathbb{R}^{2d}}\hspace{-3mm}ug(dy,du)\Big|\Big]\\
\leq&C_R\int_{\mathbb{R}^{4d}}\hspace{-3mm}\big|v^2-u^2\big|d\pi(x,v,y,u)+C_R\int_{\mathbb{R}^{4d}}\hspace{-3mm}\big|v-u\big|d\pi(x,v,y,u)\\
\leq&C_R \mathcal{W}_1(f,g),
\end{split}
\end{align*}
where the second inequality follows from that $h(r)=r^\frac{2\alpha-1}{2}$ is a concave function when $1\leq\alpha<3/2$.

Combining the above three estimates together, we know that
 $$
 \big|H_{[f]}(x,v)-H_{[g]}(x,v)\big|\leq C_R\mathcal{W}_1(f,g).
 $$
\end{proof}

\begin{lemma} \label{l4}
The following statements for Monge-Kantorovich-Rubinstein distance hold(the proofs can be found in \cite{WELL} Lemma 3.11 and Lemma 3.13).
 \begin{enumerate}
 \item If $\mathcal{T}_1,~\mathcal{T}_2:\mathbb{R}^d\rightarrow\mathbb{R}^d$ are two Borel measurable maps, and $f\in\mathcal{P}_1(\mathbb{R}^{d})$, then
$$
\mathcal {W}_1(\mathcal{T}_1\#f,\mathcal{T}_2\#f)\leq\|\mathcal{T}_1-\mathcal{T}_2\|_{{L^\infty}({\rm supp} f)}.
$$
 \item If $\mathcal{T}:\mathbb{R}^d\rightarrow\mathbb{R}^d$ is a locally Lipchitz map with the Lipchitz constant $L_R$ on a ball $B_R$, and $f, g\in\mathcal{P}_1(\mathbb{R}^{d})$ have compact supports in $B_R$, then
$$
\mathcal {W}_1(\mathcal{T}\#f,\mathcal{T}\#g)\leq L_R \mathcal{W}_1(f,g).
$$
 \end{enumerate}
\end{lemma}

With the above preparations, we will prove the local existence and uniqueness of measure valued solution by Banach fixed point theorem.
\begin{proof}
{\it of theorem \ref{thmlocal}}.
Let
$$
\mathcal{M}_T=\Big\{f\in  C\big([0, T];\mathcal{P}(\mathbb{R}^{d}\times\mathbb{R}^{d})\big), \mbox{ supp} f\subset{B_R} \mbox{ for all } t\in[0,T],  R=2R_0\Big\},
$$
where $T$ will be chosen later. Define the distance in $\mathcal{M}_T$ by $$\mathcal {W}^{*}_1(f,g):=\sup\limits_{t\in [0,T]}\mathcal {W}_1(f,g),$$
then $\mathcal{M}_T$ is a complete metric space. Consider the following map
$$
\mathcal{S}(f):={\mathcal{T}}^t_{H_{[f]}}\#f_0, \qquad \forall f\in\mathcal{M}_T,
$$
then a fixed point of $\mathcal{S}$ is a solution of equation \eqref{kineticequs} in $[0,T]$. We are left to prove that there exists $T>0$ such that $\mathcal{S} (\mathcal{M}_T)\subset \mathcal{M}_T$ and $\mathcal{S}$ is a contraction. We will prove them in the following two steps.

{\bf Step 1.}  $\mathcal{S}(\mathcal{M}_T)\subset \mathcal{M}_T$ for small $T>0$.

From Lemma \ref{Hlp} and Corollary \ref{Fl}, we have
 $$
 |H_{[f]}(t,x,v)|\leq C_R, \mbox{ for all } (t,x,v)\in [0,T]\times B_R
  $$
and
  $$
  \big|\frac{d}{dt}\mathcal{T}^t_{H_{[f]}}(x_0,v_0)\big|\leq C_R \mbox{   for all } (x_0,v_0)\in{B_{R_0}}.
  $$
Since the initial probability measure $f_0$ has support contained in $B_{R_0}$, we can choose $T<\frac{R_0}{C_R}$, so that ${\mathcal{T}}^t_{H_{[f]}}\#f_0$ has support in $B_R$, $R=2R_0$ for all $t\in[0,T].$

The continuity of $\mathcal{S}(f)$ in $t$ can be obtained by Corollary  \ref{Fl}, i.e. , for $(x_0,v_0)\in{B_{R_0}}$, and $s, t\in[0,T]$,
\begin{eqnarray*}
\big|\mathcal{T}^t_{H_{[f]}}(x_0,v_0)-\mathcal{T}^s_{H_{[f]}}(x_0,v_0)\big|
&=&\big|\int_{0}^t\psi_{H_{[f]}}(\xi,x,v)d\xi-\int_{0}^s\psi_{H_{[f]}}(\xi,x,v)d\xi\big|\\
&\leq&\big|\int_{s}^t\psi_{H_{[f]}}(\xi,x,v)d\xi\big| \leq C_R|t-s|.
\end{eqnarray*}

{\bf Step 2.} $\mathcal{S}$ is a contraction in $\mathcal{M}_T$ for small $T>0$.

For two functions $f, g\in\mathcal{M}_T$, the distance between their images is
\begin{align*}
  \mathcal {W}^{*}_1\big(\mathcal{S}(f), \mathcal{S}(g)\big)=\sup\limits_{t\in{[0,T]}}\mathcal {W}_1\big({\mathcal{T}}^t_{H_{[f]}}\#f_0,{\mathcal{T}}^t_{H_{[g]}}\#f_0\big),
\end{align*}
which can be estimated by using Lemma \ref{l4}, Lemma \ref{p1}  and Lemma \ref{l5} step by step, namely, $\forall t\in[0,T]$, we have
\begin{align*}
\mathcal {W}_1\big({\mathcal{T}}^t_{H_{[f]}}\#f_0,{\mathcal{T}}^t_{H_{[g]}}\#f_0\big)
\leq&\big\|{\mathcal{T}}^t_{H_{[f]}}-{\mathcal{T}}^t_{H_{[g]}}\big\|_{L^{\infty}(B_{R_0})}\\
\leq&\frac{e^{(1+L_R)t}-1}{(1+L_R)}\sup\limits_{t\in{[0,T]}}\big\|H_{[f]}-H_{[g]}\big\|_{L^{\infty}(B_R)}\\
\leq&\frac{e^{(1+L_R)T}-1}{(1+L_R)}C_R\sup\limits_{t\in{[0,T]}}\mathcal {W}_1(f,g)\\
\leq&\frac{e^{(1+L_R)T}-1}{(1+L_R)}C_R\mathcal {W}^{*}_1(f,g).
\end{align*}
Now we can choose $T>0$ (a time smaller than the one in Step 1) such that $\mathcal{S}$ is a contraction.
\end{proof}

\section{A priori estimates and global wellposedness}

The local wellposedness implies that the solution exists within time interval $[0,T^*)$ with $T^*\leq +\infty$. If $T^*=+\infty$, the solution exists globally. On the other hand, if $T^*<+\infty$, then the solution has no compact support when $t\rightarrow T^*-0$. In this section, we will give the a priori estimates within the interval $[0,T^*)$ and more importantly, we will precisely estimate the particle trajectory. These further estimates imply the global existence, i.e. $T^*=+\infty$. The key ingredient here is the estimates on second moments both in space and velocity.

\subsection{A priori estimates}
We will start from the weak formulation of the equation \eqref{weaksystem} and give a series of estimates in this subsection. Here two quantities are very important,
\begin{eqnarray}
\label{G}\mathcal {G}[f(t)]=\int_{\mathbb{R}^{2d}}|v|^2f(dx,dv)-\big(\int_{\mathbb{R}^{2d}}vf(dx,dv)\big)^2,\\
\label{Gamma}\Gamma[f(t)]=\int_{\mathbb{R}^{2d}}|x|^2f(dx,dv)-\big(\int_{\mathbb{R}^{2d}}xf(dx,dv)\big)^2.
\end{eqnarray}

All the estimates in this subsection are obtained for any measure valued solution of \eqref{kineticequs} in the space $C\big([0,T^*);\mathcal{P}_2(\mathbb{R}^d\times \mathbb{R}^d)\big)$ with
\begin{align*}
\mathcal{P}_2(\mathbb{R}^d\times \mathbb{R}^d)=\{f\in \mathcal{P}; \int_{\mathbb{R}^{2d}}(x^2+v^2)f (dx,dv) <+\infty\}.
\end{align*}
Actually, the measure valued solution we have obtained before has always compact support within its existence time interval. Therefore, we can simply apply the following estimates on it.

\begin{lemma} \label{lemderivatives}
 Assume $G(v)=-G(-v)$ for any $v\in\mathbb{R}^d$. If $f\in C\big([0,T^*);\mathcal{P}_2(\mathbb{R}^d\times \mathbb{R}^d)\big)$ is a measure valued solution of equation \eqref{kineticequs}, then the following equations hold
\begin{align*}
\dfrac{d}{dt}\int_{\mathbb{R}^{2d}}vf(t,dx,dv)=&0,\\
\dfrac{d}{dt}\int_{\mathbb{R}^{2d}}xf(t,dx,dv) =& \int_{\mathbb{R}^{2d}}vf(t,dx,dv),\\
\displaystyle\frac{d}{dt}\int_{\mathbb{R}^{2d}}|x|^2f(t,dx,dv)=&2\int_{\mathbb{R}^{2d}}xvf(t,dx,dv),
\end{align*}
and
\begin{align*}
&\displaystyle\frac{d}{dt}\int_{\mathbb{R}^{2d}}|v|^2f(t,dx,dv)\\
=& -\int_{\mathbb{R}^{4d}}\Phi(|x-y|)G(v-v_{*})\cdot(v-v_{*})f(t,dx,dv)f(t,dy,dv_{*})\\
& +\displaystyle\mathcal {G}^\frac{2\alpha-1}{2}[f(t)]\int_{\mathbb{R}^{4d}}F(|x-y|^2)(v-v_{*})\cdot(x-y)f(t,dx,dv)f(t,dy,dv_{*}).
\end{align*}
Furthermore,
\begin{align*}
\mathcal {V}_1(t):=&\int_{\mathbb{R}^{2d}}vf(t,dx,dv)\equiv\mathcal {V}_1(0);\\
\mathcal {X}_1(t):=&\int_{\mathbb{R}^{2d}}x f(t,dx,dv)\equiv\mathcal {V}_1(0)t+\int_{\mathbb{R}^{2d}}x f_0(dx,dv).
\end{align*}
\end{lemma}

\begin{proof}
From the weak formulation of the kinetic system, i.e. \eqref{weaksystem}, we know
  \begin{align*}
&\frac{d}{dt}\int_{\mathbb{R}^{2d}}vf(t,dx,dv)\\
=&-\int_{\mathbb{R}^{4d}}\Phi(|x-y|)G(v-v_{*})f(t,dx,dv)f(t,dy,dv_{*})\\
&+\mathcal {G}^\frac{2\alpha-1}{2}[f(t)]\int_{\mathbb{R}^{4d}}F(|x-y|^2)(x-y)f(t,dx,dv)f(t,dy,dv_{*})\\
=&0,
\end{align*}
where we have used antisymmetry of these two integrals, namely interchanging the variables $(x,v)$ and $(y,v_{*})$ in the integrals.
We can use the similar discussion in getting the last equation, i.e.
\begin{align*}
& \frac{d}{dt}\int_{\mathbb{R}^{2d}}|v|^2f(t,dx,dv)\\
=&-2\int_{\mathbb{R}^{4d}}\Phi(|x-y|)G(v-v_{*})\cdot vf(t,dx,dv)f(t,dy,dv_{*})\\
&+2\mathcal {G}^\frac{2\alpha-1}{2}[f(t)]\int_{\mathbb{R}^{4d}} F(|x-y|^2)(x-y)\cdot vf(t,dx,dv)f(t,dy,dv_{*})\\
=&-\int_{\mathbb{R}^{4d}}\Phi(|x-y|)G(v-v_{*})\cdot(v-v_{*})f(t,dx,dv)f(t,dy,dv_{*})\\
&+\mathcal {G}^\frac{2\alpha-1}{2}[f(t)]\int_{\mathbb{R}^{4d}}F(|x-y|^2)(x-y)\cdot(v-v_{*})f(t,dx,dv)f(t,dy,dv_{*}).
\end{align*}
The second and third statements can be obtained directly by using \eqref{weaksystem}.
\end{proof}

\begin{lemma}[A priori estimate for asymptotic flocking]\label{lemflocking}
Under the assumptions \eqref{assum2} \eqref{assum3}, if $f\in C\big([0,T^*);\mathcal{P}_2(\mathbb{R}^d\times \mathbb{R}^d)\big)$ is a measure valued solution of equation \eqref{kineticequs}, then
\begin{eqnarray}
\label{estG}\mathcal {G}[f(t)]\leq\left\{\begin{array}{ll}\mathcal {G}[f_0]e^{-C^{*}t}& \alpha =1\\
                         \big(\mathcal {G}^{1-\alpha}[f_0] +(\alpha-1)C^*t \big)^{-\frac{1}{\alpha-1}} & 1<\alpha<3/2 \end{array}\right.
\end{eqnarray}
and for $1\leq \alpha<5/4$ there exists a constant $C$ which is independent of $t$ such that
\begin{eqnarray}
\label{estGamma} \Gamma[f(t)]\leq C.
\end{eqnarray}
\end{lemma}

\begin{proof}
From Lemma \ref{lemderivatives}, we get that $\mathcal {G}[f(t)]=\int_{\mathbb{R}^{2d}}|v-\mathcal{V}_1(0)|^2f(t,dx,dv)$ via the following computations
\begin{align*}
\mathcal {G}[f(t)]
=&\int_{\mathbb{R}^{2d}}|v|^2f(dx,dv)-\big(\mathcal{V}_1(0)\big)^2\\
=&\int_{\mathbb{R}^{2d}}|v|^2f(dx,dv)-2{\mathcal{V}_1(0)}\int_{\mathbb{R}^{2d}}vf(dx,dv)+\big(\mathcal{V}_1(0)\big)^2\\
=&\int_{\mathbb{R}^{2d}}\big|v^2-2\mathcal{V}_1(0)v+\big(\mathcal{V}_1(0)\big)^2\big|f(dx,dv)\\
=&\int_{\mathbb{R}^{2d}}|v-{\mathcal{V}_1(0)}|^2f(dx,dv).
\end{align*}
Using Lemma \ref{lemderivatives} and assumptions \eqref{assum2} \eqref{assum3}, we have
\begin{align*}
&\frac{d}{dt}\mathcal {G}[f(t)]
=\frac{d}{dt}\int_{\mathbb{R}^{2d}}|v|^2f(dx,dv)\\
=& -\int_{\mathbb{R}^{4d}}\Phi(|x-y|)G(v-v_{*})\cdot(v-v_{*})f(t,dx,dv)f(t,dy,dv_{*})\\
& +\displaystyle\mathcal {G}^\frac{2\alpha-1}{2}[f(t)]\int_{\mathbb{R}^{4d}}F(|x-y|^2)(v-v_{*})\cdot(x-y)f(t,dx,dv)f(t,dy,dv_{*})\\
\leq&-\Phi^{*}G^*\int_{\mathbb{R}^{4d}}|v-v_{*}|^{2\alpha}f(t,dx,dv)f(t,dy,dv_{*})\\
& +F^* \displaystyle\mathcal {G}^\frac{2\alpha-1}{2}[f(t)]\int_{\mathbb{R}^{4d}}|v-v_{*}|f(t,dx,dv)f(t,dy,dv_{*})\\
\leq&-\Phi^{*}G^*\Big(\int_{\mathbb{R}^{4d}}|v-v_{*}|^2f(t,dx,dv)f(t,dy,dv_{*})\Big)^\alpha\\
& +F^* \displaystyle\mathcal {G}^\frac{2\alpha-1}{2}[f(t)]\Big(\int_{\mathbb{R}^{4d}}|v-v_{*}|^2f(t,dx,dv)f(t,dy,dv_{*})\Big)^\frac{1}{2}.
\end{align*}
Notice that
\begin{eqnarray*}
2\mathcal {G}[f(t)] =\int_{\mathbb{R}^{4d}}|v-v_{*}|^2f(t,dx,dv)f(t,dy,dv_{*}).
\end{eqnarray*}
Therefore, the inequality for $\mathcal {G}[f(t)] $ is
\begin{eqnarray*}
\frac{d}{dt}\mathcal {G}[f(t)] \leq  -(2^\alpha \Phi^* G^* -\sqrt{2} F^*) \mathcal {G}^\alpha [f(t)] :=-C^*\mathcal {G}^\alpha [f(t)],
\end{eqnarray*}
which implies that $\mathcal {G}[f(t)]$ decay in time. If $\mathcal {G}[f_0]=0$, then we always have $\mathcal {G}[f(t)]=0$. Otherwise, in the case of $\alpha=1$, we have
\begin{eqnarray*}
\mathcal {G}[f(t)]\leq\mathcal {G}[f_0]e^{-C^{*}t}.
\end{eqnarray*}
In the case of $1<\alpha<3/2$, we have
\begin{eqnarray*}
\mathcal {G}[f(t)]\leq\big(\mathcal {G}^{1-\alpha}[f_0] +(\alpha-1)C^*t \big)^{-\frac{1}{\alpha-1}} :=(B+At)^{-\frac{1}{\alpha-1}}.
\end{eqnarray*}
Hence, \eqref{estG} is proved.

Using Lemma \ref{lemderivatives}
\begin{align*}
\frac{d}{dt}\Gamma[f(t)]
=&\frac{d}{dt}\Big[\int_{\mathbb{R}^{2d}}|x|^2f(dx,dv)-\big(\int_{\mathbb{R}^{2d}}xf(dx,dv)\big)^2\Big]\\
=&2\int_{\mathbb{R}^{2d}}x\big(v-\mathcal{V}_1(0)\big)f(t,dx,dv)\\
\leq&\Big[\int_{\mathbb{R}^{2d}}|v-\mathcal{V}_1(0)|^2f(t,dx,dv)\Big]^{\frac{1}{2}}\Big[\int_{\mathbb{R}^{2d}}|x|^2f(t,dx,dv)\Big]^{\frac{1}{2}}\\
=&\mathcal {G}^{\frac{1}{2}}[f(t)]\Big[\Gamma[f(t)]+\big(\mathcal{V}_1(0)t+\int_{\mathbb{R}^{2d}}x f_0(dx,dv)\big)^2\Big]^{\frac{1}{2}}\\
\leq&\mathcal {G}^{\frac{1}{2}}[f(t)]\Big[\Gamma^{\frac{1}{2}}[f(t)]+\mid\mathcal{V}_1(0)\mid t+\big|\int_{\mathbb{R}^{2d}}x f_0(dx,dv)\big|\Big].
\end{align*}
If $\mathcal {G}[f_0]=0$, then $\frac{d}{dt}\Gamma[f(t)]\leq0$ which means that $\Gamma[f(t)]\leq \Gamma[f_0]$. If $\mathcal {G}[f_0]\neq 0$, then $B\neq 0$, and we have
\begin{eqnarray*}
\frac{d}{dt}\Gamma[f(t)]
&\leq & \mathcal {G}^{\frac{1}{2}}[f(t)]\dfrac{B +At}{B}\Big[\Gamma^{\frac{1}{2}}[f(t)]+\mid\mathcal{V}_1(0)\mid \dfrac{B}{A}+\big|\int_{\mathbb{R}^{2d}}x f_0(dx,dv)\big|\Big]
\end{eqnarray*}
Let $\frac{B}{A}\mid\mathcal{V}_1(0)\mid+\big|\displaystyle\int_{\mathbb{R}^{2d}}x f_0(dx,dv)\big|=a$, we get
$$
\frac{\frac{d}{dt}\Gamma[f(t)]}{\Gamma^{\frac{1}{2}}[f(t)]+a}\leq \mathcal {G}^{\frac{1}{2}}[f(t)]\dfrac{B +At}{B}
$$
Therefore, when $1\leq \alpha<\frac{5}{4}$, by using \eqref{estG}, we have
\begin{eqnarray*}
&&\Gamma^{\frac{1}{2}}[f(t)]-a \ln\big(\Gamma^{\frac{1}{2}}[f(t)]+a\big)\\
&\leq&\Gamma^{\frac{1}{2}}[f_0]-a\ln\big(\Gamma^{\frac{1}{2}}[f_0]+a\big)+\frac{1}{2}\displaystyle\int^t_0\mathcal {G}^{\frac{1}{2}}[f(s)]\dfrac{B +As}{B}ds
\leq C,
\end{eqnarray*}
which implies there exists a constant $C>0$ independent of $t$ such that
\begin{eqnarray*}
\Gamma[f(t)] \leq C.
\end{eqnarray*}
\end{proof}

\subsection{Estimates on characteristics and global existence}
In this part, we finish the proof of global existence by estimating the characteristics and the growing speed of compact support.

\begin{lemma} \label{lemRt}
If supp$f_0\subset B_{R_0}$ and the assumptions \eqref{assum1}-\eqref{assum4} hold, $\alpha\in [1,5/4)$, let $f(t,x,v)$ be the measure valued solution of \eqref{kineticequs} within $[0,T^*)$, then ${\rm supp}f(t,\cdot,\cdot)\subset B_{R(t)}$ for $0\leq t<T^*$ with
\begin{eqnarray*}
R(t)\leq R_0e^{Ct}+C(e^{Ct}-1)^\frac{1}{2},
\end{eqnarray*}
where $C$ depends only on $R_0$.
\end{lemma}
\begin{proof} We will use within this proof $C$ to be universal constant depends only on $R_0$.

According to the a priori estimates we have obtained before, i.e. Lemma \ref{lemderivatives} and Lemma \ref{lemflocking}, we have
\begin{eqnarray}
\nonumber&\mathcal {V}_2(t):=\int_{\mathbb{R}^{2d}}|v|^2f(t,dx,dv)\leq\mathcal {V}_2(0);&\\
\label{allest}&\mathcal {X}_2(t):=\int_{\mathbb{R}^{2d}}|x|^2 f(t,dx,dv)\leq C+\big(\int_{\mathbb{R}^{2d}}x f(t,dx,dv)\big)^2;&\\
\nonumber &\int_{\mathbb{R}^{2d}}|v|f(t,dx,dv)\leq\Big(\int_{\mathbb{R}^{2d}}|v|^2f(t,dx,dv)\Big)^{\frac{1}{2}}\leq\mathcal {V}^{\frac{1}{2}} _2(0);&\\
\nonumber &\int_{\mathbb{R}^{2d}}|x| f(t,dx,dv)\leq\Big(\int_{\mathbb{R}^{2d}}|x|^2 f(t,dx,dv)\Big)^{\frac{1}{2}} \leq C+\big|\mathcal {V}_1(0)\big|t,&
\end{eqnarray}
where C depends on $\mathcal {X}_1(0)$ and $\mathcal {X}_2(0)$, thus only depends on $R_0$.

Now we consider the ODE system for characteristics
\begin{equation}\label{odecharac}
  \left\{\begin{array}{l}
  \frac{d}{dt}X=V,\\
  \frac{d}{dt}V=H_{[f]}(t,X,V).
  \end{array}\right.
\end{equation}

With the help of assumption \eqref{assum4} and all the estimates in \eqref{allest}, we have
\begin{align*}
&H_{[f]}(t,X,V)\\
\leq & C\Big(\int_{\mathbb{R}^{2d}}(|X-x|+|V-v|+1)f(dx,dv)+\big(\int_{\mathbb{R}^{2d}}\!\!|v|^2f(dx,dv)\\
&\qquad -\big(\int_{\mathbb{R}^{2d}}\!\!vf(dx,dv)\big)^2\big)^\frac{2\alpha-1}{2}\Big)\\
\leq &C\Big(\big(\int_{\mathbb{R}^{2d}}(|x|^2+|v|^2)f(dx,dv)\big)^\frac{1}{2}+\mathcal{V}^{\frac{2\alpha-1}{2}} _2(0)+1+|X|+|V|\Big)\\
\leq& C\big(1+|X|+|V|+t\big),
\end{align*}
where the constant $C$ depends only on $\mathcal{X}_1(0),\mathcal{X}_2(0)$,$\mathcal{V}_1(0)$ and $\mathcal{V}_2(0)$, namely depending on $R_0$.

Therefore, the characteristic system gives the following estimates on the particle trajectory
\begin{equation}\label{equs}
  \left\{\begin{array}{l}
  \frac{d{X^2}}{dt}\leq \frac{1}{2}(X^2+V^2),\\
  \frac{d{V^2}}{dt}\leq C \big[(X^2+V^2)+(1+t)^2\big]
  \end{array}\right.
\end{equation}
Consequently,
$$
X^2(t)+V^2(t)\leq \Big(X^2(0)+V^2(0)\Big)e^{Ct}+C(e^{Ct}-1),
$$
where $C$ depends only on $R_0$, which means that the growth of characteristics is bounded in time.
\end{proof}

\begin{remark}\label{remonref3}
The technique used to prove the well posedness of the system in our article is very similar to the one used in the reference \cite{WELL}. However, $C$ depends only on $R_0$ is the biggest difference in estimating the characteristics, which is indispensable in proving the global existence, and this key step was missing in \cite{WELL}. In fact, in \cite{WELL} Lemma 3.11, the constant $C$ depends not only on $R_0$ but also on $C_E$ from Hypothesis 3.1, consequently $C$ depends on the first moment of measure valued solution $f$, i.e. $\int_{\mathbb{R}^{2d}}|x|f(t,dx,dv)$. Furthermore, we cannot get the result that the growth of characteristics is bounded without estimating the first moment, which means that the \emph{a priori} estimates we have obtained in Lemma \ref{lemderivatives} and Lemma \ref{lemflocking} are necessary.
\end{remark}

\section{Stability and mean field limit}
In this part, we will prove the stability of the measure valued solution. Moreover, a direct consequence of stability is the so called mean field limit of the complex system. More precisely, if the initial data is an approximation of distributions of finite number of particles, then as the number of particles goes to infinity, the distribution of these particle trajectories converges to the measure valued solution of the kinetic equation. We will use the Monge-Kantorovich-Rubinstein distance for the mean field limit.

We point out here that with all the estimates in hands we can follow similar arguments as in \cite{WELL} to derive stability and mean field limit. For completeness, we still give a full detailed proof in this section.

\subsection{Stability}

As a preparation, we first give a stability result on the characteristic system.
\begin{lemma} \label{l7}
Let $f$ be the measure valued solution of \eqref{kineticequs}. For any given $(x_{01},v_{01})$, $(x_{02},v_{02})\in B_{R_0}$, let the particle trajectories $\mathcal{T}^t_{H_{[f]}}(x_{01},v_{01}), ~\mathcal{T}^t_{H_{[f]}}(x_{02},v_{02})\in B_{R(t)}$, then
$$
|\mathcal{T}^t_{H_{[f]}}(x_{01},v_{01})-\mathcal{T}^t_{H_{[f]}}(x_{02},v_{02})|\leq|(x_{01},v_{01})-(x_{02},v_{02})| \exp{\hspace{-1mm}\big\{\int_0^t\hspace{-1mm}(L_{R(s)}+1)ds\big\}},
$$
where $L_{R(t)}$ is the Lipschitz constant of $H_{[f]}(t,X,V)$ in $B_{R(t)}$.
\end{lemma}
\begin{proof} The ODE system for characteristic is
\begin{equation}\nonumber
  \left\{\begin{array}{ll}
   \frac{d}{dt}\mathcal{T}^t_{H_{[f]}}(x_{0i},v_{0i})=\Psi_{H_{[f]}}\big(t,x_i(t),v_i(t)\big),&\\
   \big(x_i(0),v_i(0)\big)=(x_{0i},v_{0i}),&i=1,2.
  \end{array}\right.
\end{equation}
By using Corollary \ref{Fl}, we get
\begin{align*}
~&\big|\mathcal{T}^t_{H_{[f]}}(x_{01},v_{01})-\mathcal{T}^t_{H_{[f]}}(x_{02},v_{02})\big|\\
\leq&\big|(x_{01},v_{01})-(x_{02},v_{02})\big|+\int_0^t\big|\psi_{H_{[f]}}\big(s,x_1(s),v_1(s)\big)-\psi_{H_{[f]}}\big(s,x_1(s),v_1(s)\big)\big|ds\\
\leq&\big|(x_{01},v_{01})-(x_{02},v_{02})\big|+\int_0^t (L_R(H_{[f]})+1)\big|\big(x_1(s),v_1(s)\big)-\big(x_1(s),v_1(s)\big)\big|ds,
\end{align*}
Then the result will be obtained directly by using Gronwall's Lemma.
\end{proof}

Now we can prove Theorem \ref{T2} in the following.
\begin{proof}
Without loss of generality, we can assume that ${\rm supp} f_0\bigcup {\rm supp}g_0 \in B_{R_0}$. Then from Lemma \ref{lemRt}, we know that
$$
{\rm supp} f(t,\cdot,\cdot)\cup {\rm supp}g(t,\cdot,\cdot) \in B_{R(t)}.
$$
Then the distance between $f$ and $g$ can be estimated by
\begin{eqnarray*}
\mathcal {W}_1(f,g)&=&\mathcal {W}_1({\mathcal{T}}^t_{H_{[f]}}\#f_0,{\mathcal{T}}^t_{H_{[g]}}\#g_0)\\
&\leq&\mathcal {W}_1({\mathcal{T}}^t_{H_{[f]}}\#f_0,{\mathcal{T}}^t_{H_{[f]}}\#g_0)+\mathcal {W}_1({\mathcal{T}}^t_{H_{[f]}}\#g_0,{\mathcal{T}}^t_{H_{[g]}}\#g_0).
\end{eqnarray*}
By using Lemma \ref{l4} and \ref{l7}, we have
\begin{eqnarray*}
&&\mathcal {W}_1({\mathcal{T}}^t_{H_{[f]}}\#f_0,{\mathcal{T}}^t_{H_{[f]}}\#g_0)\leq \exp{\big\{\int_0^t(L_{R(s)}+1)ds\big\}}\mathcal {W}_1(f_0,g_0),
\end{eqnarray*}
and by using Lemma \ref{l4}, \ref{l5} and Lemma \ref{p1}, we have
\begin{eqnarray*}
\mathcal {W}_1({\mathcal{T}}^t_{H_{[f]}}\#g_0,{\mathcal{T}}^t_{H_{[g]}}\#g_0)
&\leq&\|{\mathcal{T}}^t_{H_{[f]}}-{\mathcal{T}}^t_{H_{[g]}}\|_{{L^\infty}(B_{R_0})}\\
&\leq&\int_0^t e^{(L_{R(s)}+1)(t-s)}\|H_{[f]}-H_{[g]}\|_{{L^\infty}(B_{R(s)})}ds\\
&\leq&C_{R(t)}\int_0^t e^{(L_{R(s)}+1)(t-s)}\mathcal {W}_1(f(s,\cdot,\cdot),g(s,\cdot,\cdot))ds.
\end{eqnarray*}
By Gronwall's Lemma, there is a constant $\lambda(t)$ such that
\begin{eqnarray}\label{sta}
\mathcal {W}_1(f,g)\leq \lambda(t)\mathcal{W}_1(f_0,g_0).
\end{eqnarray}
\end{proof}

\subsection{Derivation of the mean-field limit(a particle method)}
For any given $f_0\in\mathcal{P}(\mathbb{R}^{d}\times\mathbb{R}^{d})$ with compact support in $B_{R_0}$, there exists
\begin{eqnarray}\label{mu0}
\mu_0^N=\frac{1}{N}\sum\limits_{i=1}^N \delta\big(x-x_i(0)\big)\otimes\delta\big(v-v_i(0)\big)
\end{eqnarray}
with a sequence of particles, $\{\big(x_i(0)$, $v_i(0)\big)\}_{i=1}^N$, such that,
$$
\lim\limits_{N\rightarrow\infty}\mathcal{W}_1(\mu_0^N,f_0)=0.
$$
From the global existence and uniqueness we obtained in the previous section, there exists a measure valued solution of \eqref{kineticequs} with initial data $\mu_0^N$. Moreover,
\begin{eqnarray}\label{mut}
\mu_t^N=\frac{1}{N}\sum\limits_{i=1}^N \delta\big(x-x_i(t)\big) \otimes\delta\big(v-v_i(t)\big),
\end{eqnarray}
where $x_i(t)$, $v_i(t)$ is a solution of
$$
  \left\{\begin{array}{l}
  x_i'(t)=v_i(t),\\
  v_i'(t)=\dfrac{1}{N}\sum\limits_{j=1}^N \Phi(|x_i-x_j|)G(v_j-v_i)+\dfrac{\Lambda(v)^{2\alpha-1}}{N}\sum\limits_{j=1}^NF(|x_i-x_j|^2)(x_i-x_j),
  \end{array}\right.
$$
where
$\Lambda(v)=\frac{1}{N}(\sum\limits_{i>j}(|v_i-v_j|^2))^{\frac{1}{2}}$, $\mu_t^N$ is a measure valued solution of \eqref{kineticequs} with initial data $\mu_0^N$.

On the other hand, we know that \eqref{kineticequs} has a unique solution $f$ with initial data $f_0$. Then the mean filed limit is directly a corollary from the stability result.
\begin{corollary}\label{cormeanf}
Assume $f_0$ has compact support in $B_{R_0}$, $\mu_0^N$, $\mu_t^N$ defined in \eqref{mu0} and \eqref{mut}, then
$$
\lim\limits_{N\rightarrow\infty}\mathcal{W}_1(\mu_t^N,f(t,x,v))=0\mbox{ for all }t\geq0,
$$
where $f(t,x,v)$ is the unique measure valued solution to equation \eqref{kineticequs} with initial data $f_0$.
\end{corollary}

\begin{remark}
The mean field limit via particle method is explained in the following diagram.
\begin{eqnarray*}
{\mu_0^N} & \stackrel{ N\rightarrow\infty}{\longrightarrow} & { f_0} \\
{\mathcal{T}}^t_{H_{[\mu_t^N]}}\#\mu_0^N \downarrow &&\downarrow {\mathcal{T}}^t_{H_{[f]}}\#f_0  \\
{\mu_t^N } & \stackrel{ N\rightarrow\infty}{\longrightarrow} & {f(t,x,v)}
\end{eqnarray*}
The stability result from \eqref{sta} implies that
$$
\mathcal{W}_1(f(t,x,v),\mu_t^N)\leq\lambda(t)\mathcal{W}_1(f_0,\mu_0^N).
$$
The mean field limit is accomplished by taking $N\rightarrow\infty$.
\end{remark}

\section{Time-Asymptotic behavior of kinetic flocking }
The global measure valued solution has long time flocking behavior. From microscopic point of view, flocking means that in the long time, all particles will go in a group together with the same velocity. The corresponding quantities to describe this phenomenon is $\mathcal{G}[f]$ and $\Gamma[f]$.
\begin{lemma}[Long time flocking]\label{lemlongtimeflocking}
Under the assumptions \eqref{assum2} \eqref{assum3}, $\alpha\in[1,5/4)$, suppose that $f\in C\big([0,+\infty);\mathcal{P}_2(\mathbb{R}^d\times \mathbb{R}^d)\big)$ is a measure valued solution of equation \eqref{kineticequs}, then
$$
\displaystyle\lim_{t\rightarrow +\infty}\int_{\mathbb{R}^{2d}}|v-\mathcal{V}_1(0)|^2f(t,dx,dv)=0,
$$
where $\mathcal{V}_1(0)$ is the initial group velocity $\displaystyle\int_{\mathbb{R}^{2d}}vf_0(dx,dv)$, and
$$
\sup\limits_{0\leq t<\infty}\Big(\int_{\mathbb{R}^{2d}}|x|^2f(dx,dv)-\big(\int_{\mathbb{R}^{2d}}xf(dx,dv)\big)^2\Big)<C_{R_0}.
$$
\end{lemma}
It is a direct consequence of the previous a priori estimates in Lemma \ref{lemflocking}.

Another asymptotic behavior of the solution is that there is no aggregation in the long time, which can be described by the $L^\infty$ estimate of the probability density and a constant velocity as time goes to infinity.

\begin{lemma} Under the same assumption as in Lemma \ref{lemlongtimeflocking}, we have
$$
\displaystyle\lim_{t\rightarrow +\infty} \mathcal{W}_1(f(t,x,v),\rho(t,x)\delta (v-\mathcal{V}_1(0))) =0.
$$
\end{lemma}

\begin{proof}
Instead of studying the $\mathcal{W}_1$ distance, we will estimate the bounded Lipschitz distance $d$ between $f(t,x,v)$ and $\rho(t,x)\delta (v-\mathcal{V}_1(0))$. The reader is referred to appendix for the definition of bounded Lipschitz distance $d$ and the equivalence of these two distances. For any $\varphi\in\mho$, ($\mho$ is the test function space in the definition of bounded Lipschitz distance $d$, see \eqref{mho}) we have
\begin{align*}
&\big|\int_{\mathbb{R}^{2d}}\varphi(x,v)f(t,dx,dv)-\int_{\mathbb{R}^{2d}}\varphi(x,v)\rho(t,dx)d \delta (v-\mathcal{V}_1(0))\big|\\
=&\big|\int_{\mathbb{R}^{2d}}\varphi(x,v)f(t,dx,dv)-\int_{\mathbb{R}^{d}}\varphi(x,\mathcal{V}_1(0))\rho(t,dx)\big|\\
=&\big|\int_{\mathbb{R}^{2d}}\big(\varphi(x,v)-\varphi(x,\mathcal{V}_1(0))\big)f(t,dx,dv)\big|\\
\leq&\int_{\mathbb{R}^{2d}}|v-\mathcal{V}_1(0)|^2f(t,dx,dv),
\end{align*}
which means that
$$
d(f(t,x,v),\rho(t,x)\delta (v-\mathcal{V}_1(0)))\leq\int_{\mathbb{R}^{2d}}|v-\mathcal{V}_1(0)|^2f(t,dx,dv).
$$
then the proof is ended by taking $t\rightarrow\infty$.
\end{proof}

\begin{lemma}
Assume the coupling function satisfies $G(v)=-G(-v)$ and $\nabla_v\!\cdot G(v)\geq0$ for any $v\in\mathbb{R}^d$, and
 $\Phi(|x|)\geq 0$. If $f$ is a classical solution to equation \eqref{kineticequs} decaying fast enough at infinity with initial data $f_0\in L^\infty\big({\mathbb{R}^d\times\mathbb{R}^d}\big)\cap L^1\big({\mathbb{R}^d\times\mathbb{R}^d}\big)$, then
$$
\big\|f(t,x,v)\big\|_{{L^\infty}({\mathbb{R}^d\times\mathbb{R}^d})}\leq C\|f_0\big\|_{{L^\infty}({\mathbb{R}^d\times\mathbb{R}^d})}.
$$
\end{lemma}

\begin{proof}
From equation \eqref{kineticequs}, $1<p<\infty$, we have
\begin{align*}
&\frac{d}{dt}\int_{\mathbb{R}^{2d}}f^{p-1}f dxdv\\
=&-p\int_{\mathbb{R}^{2d}}f^{p-1}\Big[v\cdot\nabla_x f+\nabla_v\cdot H_{[f]}f\Big]dxdv\\
=&\int_{\mathbb{R}^{2d}}\Big[v\cdot\nabla_x f^p+\nabla_v\cdot (H_{[f]}f^p)+(p-1)(\nabla_v\cdot H_{[f]})f^p\Big]dxdv\\
=&-(p-1)\int_{\mathbb{R}^{4d}}\Phi(|x-y|)\nabla_v\cdot G(v-v_{*})f^p(x,v)f(y,v_{*})dydv_{*}dxdv,
\end{align*}
Since $ \nabla_v\!\cdot \!G(v)\geq0$, $f(t,x,v)\geq 0$, we have
$$
\frac{d}{dt}\int_{\mathbb{R}^{2d}}f^p(t,x,v)dxdv\leq 0.
$$
Or equivalently,
$$
\int_{\mathbb{R}^{2d}}f^p(t,x,v)dxdv\leq\int_{\mathbb{R}^{2d}}f^p_0(x,v)dxdv.
$$
Then interpolation implies
$$
\big\|f(t,x,v)\big\|_{{L^\infty}\big({\mathbb{R}^d\times\mathbb{R}^d}\big)}\leq C\|f_0\|_{L^\infty}.
$$

\end{proof}

\section*{Appendix}
In this appendix, for the convenience of the readers, we list here main tools about optimal transport which will be useful for the proving of well-posedness theory. Consider the space of probability measure $\mathcal{P}_1(\mathbb{R}^{d}),$
\begin{eqnarray*}
\mathcal{P}_1(\mathbb{R}^{d})&=&\Big\{f|~f \mbox{ is a probability measure on } \mathbb{R}^{d} \mbox{ and } \int_{\mathbb{R}^{d}}|x|f(dx)
<+\infty\Big\}.
\end{eqnarray*}
In $\mathcal{P}_1(\mathbb{R}^{d})$, the Monge-Kantorovich-Rubinstein distance is defined by
$$
\mathcal {W}_1(f,~g)=\inf_{\pi\in\Lambda(f,~g)}\Big\{\int_{\mathbb{R}^{d}\times\mathbb{R}^{d}}|x-y|d\pi(x,y)\Big\}
$$
where $\Lambda(f,~g)$ is the set of transference plans between the measures $f $ and $g$ which have marginals $f $ and $g$, the set is always nonempty.

Another useful tool is the bounded Lipschitz distance. We will give the definition in following.
First, define the admissible set $\mho$ of test functions:
\begin{eqnarray}\label{mho}
\mho:=\big\{\varphi:\mathbb{R}^{2d}\rightarrow\mathbb{R}, \|\varphi\|_{L^\infty}\leq1, \mbox{ Lip }(\varphi):=\displaystyle\sup_{x_1\neq x_2 \in\mathbb{R}^{2d}}\frac{|\varphi(x_1)-\varphi(x_2)|}{|x_1-x_2|}\leq1\big\}.
\end{eqnarray}
Let $\mu$, $\nu$ $\in\mathbb{M}$ be two Radon measures. Then the bounded Lipschitz distance $d(\mu$, $\nu)$ is given by
$$
d(\mu,\nu):=\displaystyle\sup_{\varphi\in\Omega}\big|\int_{\mathbb{R}^{2d}}\varphi(x,v)d\mu -\int_{\mathbb{R}^{2d}}\varphi(x,v)d\nu\big|.
$$
The bounded Lipschitz distance $d$ is equivalent to the Monge- Kantorovich- Rubinstein distance $\mathcal {W}_1$ $($see\cite {SJ}$)$.

\section*{Acknowledgments} We would like to thank Jian-Guo Liu for introducing us into this field and with many helpful discussions.


\end{document}